\documentclass[%
 reprint,
superscriptaddress,
nofootinbib,
 amsmath,amssymb,
 aps,
]{revtex4-2}
\usepackage{ulem}
\usepackage{graphicx}
\usepackage{dcolumn}
\usepackage{bm}
\usepackage[colorlinks=true,linkcolor=blue,citecolor=magenta,urlcolor=blue]{hyperref}
\usepackage{physics}
\usepackage{amsthm}
\usepackage{mathrsfs}
\usepackage{mathtools}
\newtheorem{theorem}{Theorem}
\newtheorem{attack}{Attack}
\newcommand\underrel[2]{\mathrel{\mathop{#2}\limits_{#1}}}
\newcommand{\stkout}[1]{\ifmmode\text{\sout{\ensuremath{#1}}}\else\sout{#1}\fi}
\usepackage[mathlines]{lineno}

\begin{document}

\preprint{APS/123-QED}

\title{Extending loophole-free nonlocal correlations to arbitrarily large distances}

\author{Anubhav Chaturvedi}
\email{anubbhav.chaturvedi@phdstud.ug.edu.pl}
\affiliation{Faculty of Applied Physics and Mathematics,
 Gda{\'n}sk University of Technology, Gabriela Narutowicza 11/12, 80-233 Gda{\'n}sk, Poland}
\affiliation{International Centre for Theory of Quantum Technologies (ICTQT), University of Gda{\'n}sk, 80-308
Gda\'nsk, Poland}
\author{Giuseppe Viola}%
\affiliation{International Centre for Theory of Quantum Technologies (ICTQT), University of Gda{\'n}sk, 80-308 Gda\'nsk, Poland}
\author{Marcin Paw{\l}owski} 
\affiliation{International Centre for Theory of Quantum Technologies (ICTQT), University of Gda{\'n}sk, 80-308 Gda\'nsk, Poland}

\date{\today}

\begin{abstract}\normalem
    One of the most striking features of quantum theory is that it allows spatially separated observers to share correlations that resist \emph{local hidden variable} (classical) explanations, a phenomenon referred to as Bell nonlocality. Besides their foundational relevance, the nonlocal correlations enable distant observers to accomplish classically inconceivable information processing and cryptographic feats such as unconditionally secure device-independent key distribution schemes. However, the distances over which nonlocal correlations can be realized in \emph{state-of-the-art} Bell experiments remain severely limited owing to the high threshold efficiencies of the detectors and the fragility of the nonlocal correlations to experimental noise. Instead of looking for quantum strategies with marginally lower threshold requirements, we exploit the properties of loophole-free nonlocal correlations, which are experimentally attainable today, albeit at short distances, to extend them over arbitrarily large distances. Specifically, we consider Bell experiments wherein the spatially separated parties randomly choose the location of their measurement devices in addition to their measurement settings. We demonstrate that when devices close to the entanglement source are perfect and witness extremal loophole-free nonlocal correlations, such correlations can be extended to devices placed arbitrarily far from the source, with \emph{almost-zero} detection efficiency and visibility. To accommodate imperfections close to the source, we demonstrate a specific analytical tradeoff: the higher the loophole-free nonlocality close to the source, the lower the threshold requirements away from the source. We utilize this analytical tradeoff paired with optimal quantum strategies to estimate the critical requirements of a measurement device placed away from the source in the simplest non-trivial scenario and formulate a versatile numerical method applicable to generic network scenarios. Our results demonstrate that the properties of already feasible short-distance loophole-free nonlocal correlations can be exploited to extend them to longer distances, enabling device-independent information processing and cryptography fuelled by long-distance loophole-free nonlocal correlations to become near-term commonplace technologies.
\end{abstract}

     \normalem                         
\maketitle
\textbf{\emph{Introduction:---}} Spatially separated observers cannot communicate faster than the speed of light, \`{a} la \emph{relativity}. However, they can share quantum correlations born of local measurements performed on entangled particles which resist \emph{local hidden variable} (classical) explanations. This phenomenon is called \emph{Bell nonlocality} \cite{Bell1964,BellNonlocality}. The nonlocal correlations enable the observers to accomplish classically inconceivable information processing and cryptographic feats such as unconditionally secure \footnote{Security guaranteed solely by the laws of physics, specifically, those of quantum mechanics.} device-independent (DI) \footnote{A device-independent inference is based exclusively on the observed measurement statistics, without requiring any characterization of the measurement devices, essentially independent of the underlying physical description.} quantum key distribution~\cite{Ekert1991,mayers1998quantum,BHK05,Acin2007,PhysRevLett.113.140501,DIQKDReview} and randomness expansion~\cite{colbeck2009quantum,pironio2010random,liu2021device,shalm2021device}. The efficacy of these applications necessitates loophole-free certification of nonlocality. The most challenging loophole impeding practical \emph{long-distance} DI cryptography is the detection loophole \cite{Pearle1970}, exploiting which a malicious adversary can fake nonlocal correlations if a sufficient fraction of the entangled particles remains undetected. 

A measurement device's detection efficiency, $\eta$, is the probability with which the device detects an incoming system emitted by the source. In Bell tests, closing the detection loophole amounts to having a detection efficiency above a characteristic threshold value, $\eta^*$, often referred to as the \emph{critical detection efficiency}, below which \emph{local hidden variable} models can simulate the considered nonlocal correlation. In symmetric Bell tests, wherein all detectors are equally inefficient, $\eta^*$ is a characteristic property of the target nonlocal correlation. For instance, in the simplest bipartite symmetric Bell scenario, quantum correlations that maximally violate the \emph{Clauser-Horne-Shimony-Holt} (CHSH) Bell inequality \cite{Clauser1969} have a critical detection efficiency, $\eta^* \approx0.828$ \cite{Garg1987}. However, the effective detection efficiency, $\eta$, depends not only on the properties of the measurement device but also on the losses incurred during transmission. In photonic Bell experiments, the effective detection efficiency decays exponentially with the length of the optical fiber, $l$, such that $\eta = \eta_0 10^{-\frac{\alpha l}{10}}$, where $\eta_0$ is the detection efficiency of the measuring apparatus due to the use of imperfect detectors, and $\alpha$ is the attenuation
coefficient typically $\approx0.2dB/km$ at a wavelength of $1550nm$ (third telecom window) \cite{OFLAA}. Therefore, the lower the critical detection efficiency of a nonlocal correlation, the further away the measurement devices can be from the source while retaining loophole-free nonlocal behavior. Consequently, the detection-loophole is practically unavoidable in photonic Bell experiments and DIQKD systems when using optical fibers of about $5km$ \cite{FakingBellNonlocalityMakarov} and $3.5km$  \cite{Zapatero2023} in length, respectively. Another crucial quantity for long-distance loophole-free Bell tests is the visibility, $\nu$, of the entangled quantum systems, which quantifies the amount of noise added during transmission and due to imperfections in the source. Analogously to the critical detection efficiency, each nonlocal quantum strategy has a characteristic threshold value of visibility, $\nu^*$, below which the consequent correlations cease to be nonlocal.

Over the years, several proposals have identified nonlocal quantum correlations with lower critical detection efficiencies. For instance, in the symmetric CHSH scenario, one can reduce it down to $\eta^*=2/3$ by using a pair of almost-product partially entangled qubits \cite{Eberhard1993}. However, this lower critical detection efficiency comes at the cost of very high susceptibility to noise with $\nu^*\approx 1$. The other proposals fall into one of the two categories, $(i.)$ the ones which increase the complexity of the quantum set-up by either utilizing entangled quantum systems of higher local dimension \cite{Massar2002, Vertesi2010,miklin2022exponentially, XuLatest}, or by increasing the number of spatially separated parties \cite{Larsson2001, CRV08, PVB12}, and $(ii.)$ the ones which invoke theoretical idealizations such as perfect detectors for a measuring party \cite{CL07, BGSS07, Garbarino10, AQCFCT12}. While the latter are clearly of little practical significance, the former necessitates more intricate state preparation procedures, which invariably lead to a higher susceptibility to noise and experimental fragility.

Consequently, long-distance loophole-free nonlocal correlations remain elusive, even in the near-term future. This is reflected in the current \emph{state-of-the-art} combinations of $(\eta,\nu)$ reported in photonic Bell experiments over distances $\leq 400 m$, $(0.774, 0.99)$~\cite{giustina2015significant}, $(0.763, 0.99)$~\cite{shalm2021device}, and $(0.8411, 0.9875)$~\cite{liu2021device}. 
On the other hand, if DI cryptography is to become a near-term commonplace technology, operationally certifiable robust nonlocal correlations must be sustained over distances orders of magnitude larger $( \gg 100km)$. Due to the sheer enormity of this gap, the traditional approach of looking for nonlocal quantum correlations with marginally lower critical detection efficiency seems futile. Instead, in this work, we exploit the properties of strong nonlocal correlations, which can be readily attained today, albeit at short distances, to extend them to arbitrarily large distances.

Specifically, we consider a generalization to the standard \emph{Bell experiments}, wherein each round, the measuring parties randomly choose the location (distance from the source) of their measurement devices in addition to their measurement settings. It then follows from relativity, and specifically from the so-called \emph{non-signaling condition}, that the behavior of any particular measurement device remains unaffected by the changes in the location of spatially separated measurement devices. Based exclusively on this \emph{relativistic fact} and the operational validity of quantum mechanics, we demonstrate that nonlocal correlations can be operationally certified arbitrarily far away from the source, which is to say, with \emph{arbitrarily inefficient} measurement devices, when devices close to the source operate flawlessly and witness extremal nonlocal correlations. We then proceed to derive an analytic trade-off \emph{specific} to the CHSH scenario: the higher loophole-free nonlocality close to the source, as measured by the violation of the CHSH inequality, the lower the threshold value for local hidden variable explanations away from the source. We utilize this trade-off to estimate the critical detection efficiency and visibility of the measurement device placed away from the source when the devices close are imperfect, thereby demonstrating the robustness of the effect. We also find optimal \emph{tilted} quantum strategies that minimize these requirements. Moreover, utilizing certifiable randomness as a measure of the nonlocal behavior of a device, we present a \emph{versatile numerical technique} based on the \emph{Nieto-Silleras} hierarchy of semi-definite programs \cite{Nieto_Silleras_2014, Bancal_2014}, to estimate the critical requirements of individual measurement devices in generic network scenarios with several spatially separated measurement devices. Finally, we discuss experimental setups utilizing relay switches to demonstrate this effect,  more complex network scenarios entailing multiple measurement devices, the possibility of DI cryptography schemes fuelled by this effect, and the key challenges that lay on the way.     

 \begin{figure}\centering
 \includegraphics[width=\columnwidth]{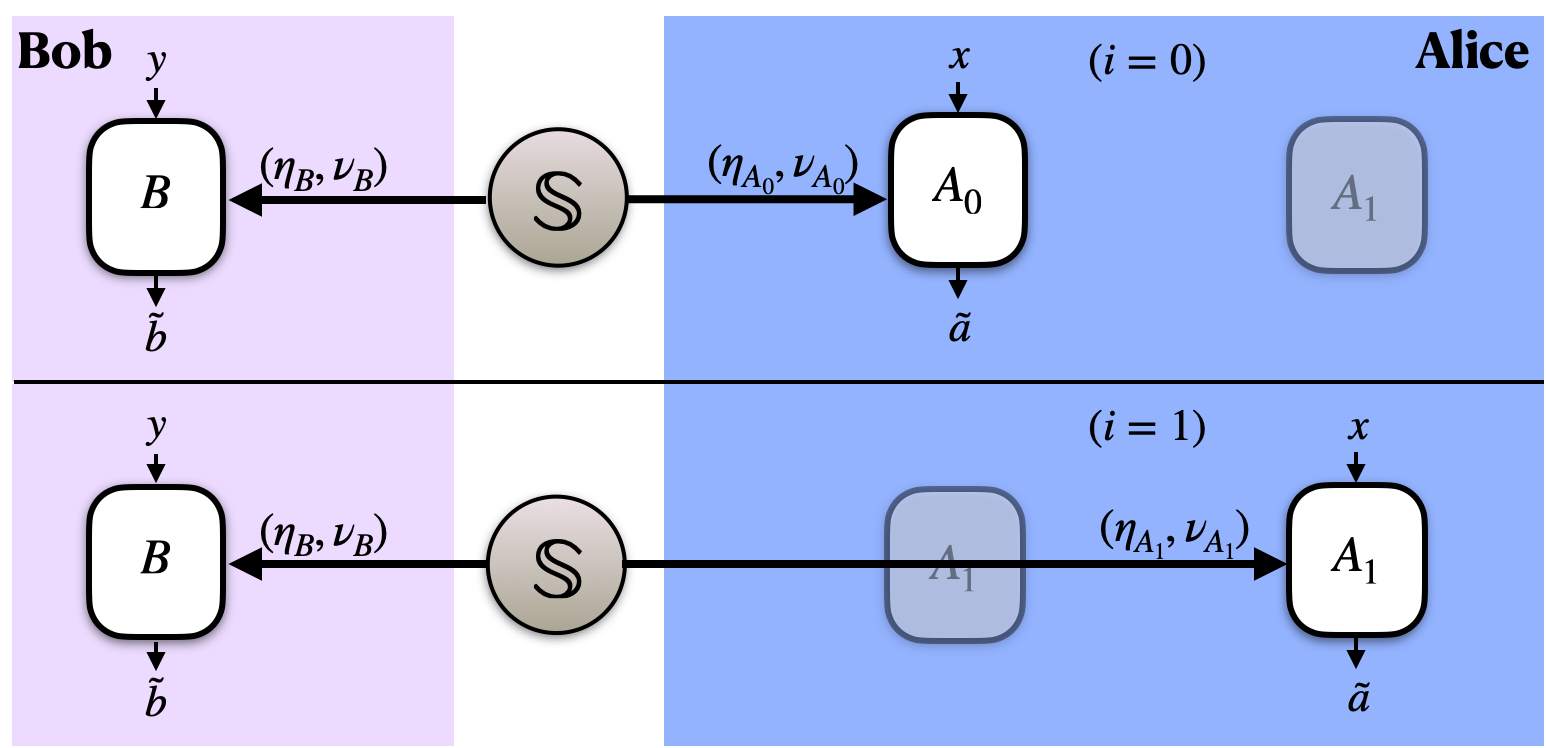}
 \caption{\label{Topology} \textbf{\emph{Topology:---}} The graphic is a schematic depiction of the generalized or 
 routed Bell CHSH experiment introduced in this work. In each round of the experiment, just as in the standard case, the parties choose their respective measurement settings, $x,y\in\{0,1\}$ and obtain outcomes, $\tilde{a},\tilde{b}\in\{+1,-1,\bot\}$, where the outcome, $\bot$, signifies the ``no-click" event. While Bob's measurement device, $B$, is at a fixed distance from the source throughout the experiment with an unvarying effective detection efficiency and visibility, $(\eta_B,\nu_B)$, the location of Alice's measurement device, $A_i$, depends on a randomly chosen input bit, $i\in\{0,1\}$. When $i=0$, Alice places her measurement device, $A_0$, close to the source with an effective detection efficiency and visibility, $(\eta_{A_0},\nu_{A_0})$, whereas when $i=1$, she places her measurement device, $A_1$, further away from the source, with a lower detection efficiency and visibility, $(\eta_{A_1},\nu_{A_1})$, such that, $\eta_{A_1}<\eta_{A_0}$, $\nu_{A_1}<\nu_{A_0}$.}
 \end{figure}
\textbf{\emph{Preliminaries:---}}
Let us consider the simplest bipartite Bell scenario entailing a source, $\mathbb{S}$, distributing \emph{entangled} quantum systems, ideally in a two-qubit pure state, $\ket{\psi}\in \mathbb{C}^2 \otimes \mathbb{C}^2$, to two spatially separated parties, Alice and Bob. The parties have measurement devices with binary inputs, $x,y\in\{0,1\}$, specifying the measurement settings, and produce binary outcomes, $a,b\in \{+1,-1\}$, respectively. In ideal circumstances, the measurement devices perform measurements corresponding to binary outcome projective observables, $\hat{a}_{x} \in B(\mathbb{C}^2),\hat{b}_{y}\in B(\mathbb{C}^2)$. The three tuple, $\mathcal{Q}\equiv \left(\ket{\psi},\{\hat{a}_{x}\}_x,\{\hat{b}_{y}\}_y\right)$, constitutes a quantum strategy (entailing operational instructions) which ideally results in the experimental behavior, $\mathbf{p}\equiv\left\{p(a,b|x,y)=\frac{1}{4}\expval{({\mathbb{I}+a\hat{a}_x})\otimes({\mathbb{I}+b\hat{b}_y})}{\psi}\right\}\in \mathbb{R}^{16}_{+}$. In general, up to local-relabelling, a given behavior, $\mathbf{p}$, is said to be \emph{nonlocal} if and only if it violates the CHSH inequality, 

\begin{align} \label{CHSH}
    C(\mathbf{p}) &\equiv  \sum_{x,y} (-1)^{x\cdot y}\expval{\hat{a}_{x}\hat{b}_{y}} \underrel{\mathcal{L}}{\leqslant} 2,
\end{align} 
where $\expval{\hat{a}_{x}\hat{b}_{y}}=\sum_{a,b}abp(a,b|x,y)$. The inequality \eqref{CHSH} holds for all behaviors, $\mathbf{p} \in\mathbb{R}^{16}_{+}$, which admit \emph{local hidden variable} explanations ($\mathcal{L}$), such that, $p(a,b|x,y){=}\sum_{\lambda\in\Lambda}p(\lambda)p^{A}_\lambda(a|x)p^{B}_{\lambda}(b|y)$, where $\lambda$ is the local hidden variable, $\Lambda$ is a measurable hidden variable state space, $p(\lambda)$ specifies the probability of the system occupying the state corresponding to $\lambda$, and for a specific $\lambda$, the conditional probability distributions, $\{p^{A}_\lambda(a|x)\}$ and $\{p^B_\lambda(b|y)\}$, represent stochastic response schemes specifying the outcome probabilities for Alice and Bob, respectively. 

However, the actual measurement devices may be imperfect and sometimes fail to detect the incoming quantum system, an event referred to as the ``no-click" event, and the effective probability with which a measurement device, $D$, ``clicks", is referred to as its detection efficiency, $\eta_D\in [0,1]$ \footnote{We assume that all detectors in a device are equally efficient.}, where, $\eta_D=1$ signifies perfect detectors . We consider a generalization of the simplest Bell experiment, depicted in FIG. \ref{Topology}, wherein Bob's measurement device, $B$, is at a fixed distance from the source throughout the experiment and has an unvarying effective detection efficiency, $\eta_B$. Alice, on the other hand, randomly chooses the spatial location of her device based on an additional input bit, $i\in\{0,1\}$. When $i=0$, she places her measurement device, $A_0$, close to the source achieving an effective detection efficiency, $\eta_{A_0}$, whereas when $i=1$, she places her device, $A_1$, further away from the source, attaining a lower effective detection efficiency, $\eta_{A_1}\leq \eta_{A_0}$. Additionally, to account for imperfections in the source and noise added during transmission from the source, we associate effective visibilities, $\nu_B,\nu_{A_i} \in [0,1]$ to each measurement device, such that the quantum state shared between the measurement devices, $(A_i,B)$, is,

\begin{align} 
    \rho(\nu_i) & =  \nu_{A_i}\nu_B\ketbra{\psi} + \nu_{A_i}(1-\nu_B) \left(\rho_B\otimes \frac{\mathbb{I}}{2}\right) \\ \nonumber
    & \hspace{0.1in} + (1-\nu_{A_i})\nu_B \left(\frac{\mathbb{I}}{2}\otimes\rho_A\right) + (1-\nu_{A_i})(1-\nu_B)\frac{\mathbb{I}_4}{4},
\end{align}

where $\rho_A=\Tr_B(\ketbra{\psi})$, $\rho_B=\Tr_A(\ketbra{\psi})$, $\mathbb{I}$ is the two-dimensional identity operator, and $\mathbb{I}_4=\mathbb{I}\otimes \mathbb{I}$. Analogously to the detection efficiencies, the effective visibility of Alice's measurement device decreases with the distance from the source, such that $\nu_{A_1}\leq \nu_{A_0}$.

Treating the no-click event as an additional outcome, $\bot$, the parties observe the experimental behavior in the form of conditional probability distributions, $\mathbf{p}_{(exp)}\equiv\{p(\Tilde{a},\Tilde{b}|(x,i),y)\} \in \mathbb{R}^{72}_{+}$, where $\Tilde{a},\Tilde{b} \in \{+1,-1,\bot\}$. A convenient way to post-process the experimental behavior, $\mathbf{p}_{(exp)}$, which avoids considering additional outcomes and as well as the fair-sampling assumption, is to assign a valid outcome, say $+1$, to each no-click event, such that, the effective distribution reduces to $\mathbf{p}_{(exp)}^{(\bot \mapsto 1)} \equiv\{p(a,b|(x,i),y)\}\in \mathbb{R}^{32}_{+}$ \cite{giustina2015significant,MarcinPostProcess}. One of the benefits of such a post-processing is the applicability of well-studied reliable means of quantifying loophole-free nonlocal correlations such as the observed value of the CHSH expression, 
\begin{equation}
    C_{A_i,B}(\mathbf{p}_{(exp)}^{(\bot \mapsto 1)}) \equiv  \sum_{x,y} (-1)^{x\cdot y}\expval{\hat{a}_{(x,i)}\hat{b}_{y}}
\end{equation}
 where $\expval{\hat{a}_{(x,i)}\hat{b}_{y}}=\sum_{a,b}abp(a,b|(x,i),y)$. 
 
 The measurement devices $(A_i,B)$ are said to share loophole-free nonlocal correlations
 if the behaviour $\mathbf{p}_{(exp)}^{(\bot \mapsto 1)}$ \emph{does not allow} for a \emph{local hidden variable} explanation of the form,
\begin{equation} \label{LHVGenBellExp}
    p(a,b|(x,i),y)\underrel{{\mathcal{L}(A_i)}}{=}\sum_{\lambda\in\Lambda}p(\lambda)p^{A}_\lambda(a|x,i)p^{B}_{\lambda}(b|y),
\end{equation}
for all $a,b\in\{+1,-1\}$ and $x,y\in\{0,1\}$, where $\lambda$ is a local hidden variable, $\Lambda$ is a measurable hidden variable state space, $p(\lambda)$ specifies the probability of the shared system occupying the state corresponding to $\lambda$, and for a specific $\lambda$, the conditional probability distributions, $\{p^{A}_\lambda(a|x,i)\}$ and $\{p^B_\lambda(b|y)\}$, represent stochastic response schemes specifying the outcome probabilities for Alice and Bob, respectively. As assigning a local pre-determined outcome to the ``no-click" event cannot increase the \emph{local hidden variable} bound, $2$, of the CHSH expression, the experimental violation of the CHSH inequality \eqref{CHSH}, $C_{A_iB}(\mathbf{p}_{(exp)}^{(\bot \mapsto 1)})> 2$, constitutes a sufficient operational condition for certifying loophole-free nonlocal correlations between the spatially separated measurement devices, $(A_i,B)$. However, as we demonstrate below, this threshold requirement for the certification of loophole-free nonlocal correlations between $(A_1,B)$ reduces drastically when $(A_0,B)$ witness loophole-free violation of the CHSH inequality \eqref{CHSH}. Specifically, let us suppose that the measurement devices close to the source $(A_0,B)$ witness a loophole-free violation of the CHSH inequality $C_{A_0B}(\mathbf{p}_{(exp)}^{(\bot \mapsto 1)})> 2$, this implies that their behaviour must spring from an underlying quantum set-up, i.e., from quantum measurements $\{M^{(x,i=0),\lambda}_a\in B_+(\mathscr{H}_{A}),M^{y,\lambda}_b \in B_+(\mathscr{H}_B)\}$ on shared entangled states $\{\rho_\lambda\in B_+(\mathscr{H}_{A}\otimes\mathscr{H}_B)\}$, where $\mathscr{H}_{A},\mathscr{H}_{B}$ are arbitrary underlying Hilbert spaces, such that, for all $a,b\in\{+1,-1\}$ and $x,y\in\{0,1\}$,
\begin{equation} \label{QuantumLHVGenBellExp}
    p(a,b|(x,i=0),y){=}\sum_{\lambda\in\Lambda}p(\lambda)\Tr(\rho_{\lambda}M^{(x,i=0),\lambda}_a\otimes M^{y,\lambda}_b).
\end{equation}
In particular, this implies that
 Bob's outcome must spring from genuine quantum measurements $\{M^{y,\lambda}_b \in B_+(\mathscr{H}_B)\}$ on his part of  entangled states $\{\rho^{(B)}_\lambda = \Tr_A\rho_\lambda\}$ shared with $A_0$, such that, $p^{B}_{\lambda}(b|y)=\Tr(\rho^{(B)}_\lambda M^{y,\lambda}_b)$. 
 
 As the source as well as the the internal workings of Bob's measurement device, $B$,  remain unaffected by the changes in the location of Alice's measurement device, i.e., by the choice $i$, if the behaviour Alice's other device were possess a \emph{local hidden variable} explanation of the form \eqref{LHVGenBellExp}, i.e., it obtains its measurement outcomes from the post-processing of a shared local hidden variable $\lambda$, then, 
\begin{align} \label{A1LHVGenBellExp2} 
    p(a,b|&(x,i=1),y)\underrel{\mathcal{L}(A_1)}{=}\\ \nonumber &\sum_{\lambda\in\Lambda}p(\lambda)p^{A}_\lambda(a|x,i=1)\Tr(\rho^{(B)}_\lambda M^{y,\lambda}_b).
\end{align}
{Notice, in the local hidden variable models invoked here, the local hidden variable $\lambda$ is implicitly assumed to be independent of the location of Alice's measurement device, i.e., $p(\lambda,i)=p(\lambda)p(i)$. As we treat $i\in\{0,1\}$ as an additional Alice's input, this assumption follows from the so-called ``measurement independence" or ``free-choice" assumption invoked while describing local hidden variable models in standard Bell experiments. In particular, this assumption has interesting cryptographic consequences, which are deferred to the discussions section, towards the end of the manuscript.}
Now, we are prepared to present our central result.

\textbf{\emph{Perfect devices close to the source:---}}
 Let us {first} consider {the} ideal case wherein the measurement devices located close to the source, $(A_0,B)$, are effectively perfect, i.e., $\eta_{B}=\eta_{A_0}=\nu_{B}=\nu_{A_0}=1$. {Then, via the following the Theorem, we demonstrate that extremal loophole-free nonlocal correlations witnessed close to source can be extended to a measurement device placed arbitrarily far away from the source,}

\begin{theorem}[\textbf{\emph{Nonlocality at arbitrary distance}}] \label{IdealCase}
If the measurement devices close to the source, $(A_0,B)$, are perfect, i.e., $\eta_{B}=\eta_{A_0}=\nu_{B}=\nu_{A_0}=1$, and witness maximally nonlocal correlations, such that, $C_{A_0B}(\mathbf{p}_{(exp)}^{(\bot \mapsto 1)})= 2\sqrt{2}$, then such nonlocal correlations may be operationally extended to Alice's other measurement device placed arbitrarily far away from the source, i.e., loophole-free nonlocal correlations {(such that there does not exist any \emph{local hidden variable} explanation of the form \eqref{A1LHVGenBellExp2})} between $(A_1,B)$ can be operationally certified for any \emph{non-zero} values of $(\eta_{A_1},\nu_{A_1})$.
\end{theorem}
\begin{proof}
    {Let the parties employ the maximally nonlocal \emph{isotropic} two qubit strategy, $\mathcal{Q}_{iso} \equiv \left(\ket{\phi^+},\{\sigma_z,\sigma_x\},\{\frac{1}{\sqrt{2}}(\sigma_z\pm\sigma_x)\}\right)$, where $\ket{\phi^+}=\frac{1}{\sqrt{2}}(\ket{00}+\ket{11})$. Consequently, the measurement devices witness the maximum quantum violation of the CHSH inequality \eqref{CHSH}, such that $C_{A_0B}(\mathbf{p}_{(exp)}^{(\bot \mapsto 1)})= 2\sqrt{2}$.}
 
 {The loophole-free observation of such extremal nonlocal quantum correlations not only discards \emph{local hidden variable} explanations of the form \eqref{LHVGenBellExp} for the measurement devices $(A_1,B)$, but also uniquely identifies the effective underlying shared quantum states $\{\rho_\lambda\}$ and the local measurements $\{M^{(x,i=0),\lambda}_a\},\{M^{y,\lambda}_b \}$ to be equivalent to the ones in $\mathcal{Q}_{iso}$, respectively, up to auxiliary degrees of freedom and local isometries for all $\lambda$, a phenomenon referred to as DI \emph{self-testing} \cite{Supic2020selftestingof}. In particular, this implies that the effective states $\{\rho^{(B)}_{\lambda}\}$ of Bob's local subsystem on which the anticommuting observables act non-trivially are equivalent to the maximally mixed qubit state, $\frac{\mathbb{I}}{2}$, such that Bob's measurement outcomes must be \emph{intrinsically random}, i.e., $\expval{\hat{b}_y}=0$ for all $y\in\{0,1\}$. As the internal workings of Bob's measurement device, $B$, remain unaffected by the changes in the location of Alice's measurement device, her effective subsystem and observables remain unaffected by Alice's choice of $i$.}

{Now, if Alice's other device, $A_1$, were to be classical, i.e., it obtains its measurement outcomes from the post-processing of a shared local hidden variable, $\lambda$, such that, the behavior witness by the measurement devices $(A_1,B)$ allows for a \emph{local hidden variable} explanation of the form \eqref{A1LHVGenBellExp2}, then the outcomes of $A_1$ cannot be correlated to that of Bob, which results in,}
{
\begin{align}\label{ClassicalCHSH} \nonumber
    {C_{A_1B}(\mathbf{p}_{(exp)}^{(\bot \mapsto 1)})}&\underrel{{\mathcal{L}(A_1)}}{=} \sum_{\lambda \in \Lambda}p(\lambda)\sum_{x,y} (-1)^{x\cdot y}\expval{\hat{a}_{(x,1)}}_\lambda\expval{\hat{b}_y}_{\lambda}, \\
    & = 0,
\end{align}}
{where $\expval{\hat{a}_{(x,1)}}_\lambda=\sum_a a p^{A}_{\lambda}(a|x,1)$, $\expval{\hat{b}_y}_\lambda=\sum_b b p^{B}_{\lambda}(b|y)$, and the second equality follows from the fact that the maximal quantum violation of the CHSH inequality, $C_{A_0B}(\mathbf{p}_{(exp)}^{(\bot \mapsto 1)})= 2\sqrt{2}$, self-tests, and hence, can only be attained by a unique quantum behavior, which in-turn implies that, $\forall \lambda : \ \expval{\hat{b}_y}_{\lambda}=\expval{\hat{b}_y} = 0$.}

{Now, if $A_1$ were to behave honestly with the strategy, $\mathcal{Q}_{iso}$, but imperfectly, i.e., with detection efficiency $\eta_{A_1}$, visibility $\nu_{A_1}$, they observe,
\begin{equation}
    {C_{A_1B}(\mathbf{p}_{(\bot \mapsto 1)})}=2\sqrt{2} \eta_{A_1}\nu_{A_1},
\end{equation}
    which violates \eqref{ClassicalCHSH}, and certifies loophole-free nonlocal correlations between $(A_1,B)$, for any non-zero values of $(\eta_{A_1},\nu_{A_1})$}

\end{proof}

Theorem \ref{IdealCase} brings forth the principal effect we employ to extend ideal loophole-free nonlocal correlations. However,
as perfect measurement devices close to the source are but a theoretical idealization, i.e., the pre-requisites of Theorem \ref{IdealCase}, namely, $\eta_{A_0}=\eta_B=\nu_0=1$, cannot be achieved in the actual experiments, we now consider cases wherein the devices close to source although better than the one placed further way, are not perfect.

\textbf{\emph{Imperfect devices close to the source:---}}  To account for such realistic cases and the estimation of the critical requirements, $(\eta^*_{A_1},\nu_{A_1}^*)$, of Alice's other measurement device, $A_1$, we present a trade-off specific to the CHSH inequality \eqref{CHSH} via the following Theorem, namely, the higher the loophole-free nonlocality witnessed close to source, the lower the threshold requirements of the measurement device placed away from the source. 

\begin{theorem}[\textbf{\emph{A specific analytical trade-off}}] \label{tradeoff}
If the measurement devices close to the source, $(A_0,B)$, witness loophole-free nonlocal correlations such that, $C_{A_0B}(\mathbf{p}_{(exp)}^{(\bot \mapsto 1)}) > 2$, then loophole-free nonlocal {(such that there does not exist any \emph{local hidden variable} explanation of the form \eqref{A1LHVGenBellExp2})} correlations between $(A_1,B)$ can be certified whenever the following inequality is violated,
\begin{equation} \label{tradeoffeq}
{C_{A_1B}(\mathbf{p}_{(exp)}^{(\bot \mapsto 1)})}\underrel{\mathcal{L}(A_1)}{\leqslant} \sqrt{8-\left(C_{A_0B}(\mathbf{p}_{(exp)}^{(\bot \mapsto 1)})\right)^2}.
\end{equation}
\end{theorem} 
\begin{proof}
Let us assume that the parties share an arbitrary behavior, $\mathbf{p} \in \mathbb{R}^{32}_{+}$. We now recall that if the measurement devices, $(A_0,B)$, witness the violation the CHSH inequality, $C_{A_0B}(\mathbf{p})>2$, then the outcomes of $B$ must be intrinsically random, which is to say, that they cannot be predicted perfectly \cite{RandomnessNonlocalityEntanglement}, specifically,
\begin{equation} \label{BoundedGuess}
 \forall y\in\{0,1\}: \  \left|\expval{\hat{b}_y}\right| \leq \frac{1}{2}{\sqrt{8-\left(C_{A_0B}(\mathbf{p})\right)^2}}.   
\end{equation}
Now, if Alice's other device, $A_1$, were to be classical, i.e., it obtains its measurement outcomes from the post-processing of a shared local hidden variable, $\lambda$, such that, the behavior witnessed by the measurement devices $(A_1,B)$ allows for a \emph{local hidden variable} explanation of the form \eqref{A1LHVGenBellExp2}, then the value of the CHSH expression can be bounded from above in the following way,
\begin{align} \label{midway}\nonumber
    C_{A_1B} &(\mathbf{p})\underrel{\mathcal{L}(A_1)}{=}  \sum_{\lambda\in\Lambda}p(\lambda)\sum_{x,y} (-1)^{x\cdot y}\expval{\hat{a}_{(x,1)}}_\lambda\expval{\hat{b}_y}_{\lambda}, \\ \nonumber
    & = \sum_{\lambda\in\Lambda}p(\lambda)\sum_x \expval{\hat{a}_{(x,1)}}_\lambda(\expval{\hat{b}_0}_{\lambda}+(-1)^x\expval{\hat{b}_1}_{\lambda}), \\ \nonumber
    & {\leqslant}  \sum_{\lambda\in\Lambda}p(\lambda)\left( \left|\expval{\hat{b}_0}_{\lambda}+\expval{\hat{b}_1}_{\lambda}\right|+\left|\expval{\hat{b}_0}_{\lambda}-\expval{\hat{b}_1}_{\lambda}\right|\right), \\ 
    & =2\sum_{\lambda\in\Lambda}p(\lambda)\max_{y\in\{0,1\}}\left|\expval{\hat{b}_y}_{\lambda}\right|
\end{align}
where the first inequality follows from the observation that $\left|\expval{\hat{a}_{(x,1)}}_\lambda\right|\leq 1$, for all $x$ and $\lambda\in\Lambda$. 

Let us now split the hidden variable state space, $\Lambda$, into the following four disjoint subspaces,
\begin{align} \nonumber
    \Lambda_0& \equiv\left\{\lambda\in\Lambda\Big|\max_{y\in\{0,1\}}\left|\expval{\hat{b}_y}_{\lambda}\right|=\expval{\hat{b}_0}_{\lambda}\right\}, \\ \nonumber
    \Lambda_1&\equiv\left\{\lambda\in\Lambda\Big|\max_{y\in\{0,1\}}\left|\expval{\hat{b}_y}_{\lambda}\right|=-\expval{\hat{b}_0}_{\lambda}\right\}, \\ \nonumber
    \Lambda_2&\equiv\left\{\lambda\in\Lambda\Big|\max_{y\in\{0,1\}}\left|\expval{\hat{b}_y}_{\lambda}\right|=\expval{\hat{b}_1}_{\lambda}\right\}, \\
    \Lambda_3&\equiv\left\{\lambda\in\Lambda\Big|\max_{y\in\{0,1\}}\left|\expval{\hat{b}_y}_{\lambda}\right|=-\expval{\hat{b}_1}_{\lambda}\right\}.
\end{align} 
This allows us expand the RHS of \eqref{midway} such that,
\begin{align} \label{midway2}\nonumber
    C_{A_1B} (\mathbf{p})&\underrel{\mathcal{L}(A_1)}{\leq} 2\Big(\sum_{\lambda\in\Lambda_0}p(\lambda)\expval{\hat{b}_0}_{\lambda} - \sum_{\lambda\in\Lambda_1}p(\lambda)\expval{\hat{b}_0}_{\lambda}  \\ \nonumber
    & \hspace{0.1in} + \sum_{\lambda\in\Lambda_2}p(\lambda)\expval{\hat{b}_1}_{\lambda} - \sum_{\lambda\in\Lambda_3}p(\lambda)\expval{\hat{b}_1}_{\lambda} \Big) \\ \nonumber
    & = 2\Big(p(\Lambda_0)\left|\expval{\hat{b}_0}_{\Lambda_0}\right|+p(\Lambda_1)\left|\expval{\hat{b}_0}_{\Lambda_1}\right| \\
    & \hspace{0.1in} + p(\Lambda_2)\left|\expval{\hat{b}_1}_{\Lambda_2}\right| + p(\Lambda_3)\left|\expval{\hat{b}_1}_{\Lambda_3}\right|\Big),
\end{align}
where $p(\Lambda_j)=\sum_{\lambda\in\Lambda_j}p(\lambda)$, and $\expval{\hat{b}_k}_{\Lambda_j}=\frac{1}{p(\Lambda_j)} \sum_{\lambda\in\Lambda_j}p(\lambda)\expval{\hat{b}_k}_{\lambda}$ for all $j\in\{0,1,2,3\}$ and $k\in\{0,1\}$. Now, we can use \eqref{BoundedGuess} to further upper bound \eqref{midway2}, such that,
\begin{align} \label{midway3} \nonumber
    C_{A_1B} (\mathbf{p})&\underrel{\mathcal{L}(A_1)}{\leq}\sum_{j\in\{0,1,2,3\}}  p(\Lambda_j)\sqrt{8-(C_{A_0B}(\mathbf{p}_{\Lambda_j}))^2}, \\ \nonumber
    & \leq \sqrt{8-\left(\sum_{j\in\{0,1,2,3\}}  p(\Lambda_j)C_{A_0B}(\mathbf{p}_{\Lambda_j})\right)^2}, \\ 
    & = \sqrt{8-\left(C_{A_0B}(\mathbf{p})\right)^2}.
\end{align}
where $\mathbf{p}_{\Lambda_j}=\frac{1}{p(\Lambda_j)}\sum_{\lambda\in\Lambda_j}p(\lambda)\mathbf{p}_{\lambda}$, and $\mathbf{p}_{\lambda} \in \mathbb{R}^{32}_+$ is the joint behavior for a specific $\lambda$, the second inequality follows from the fact that the function, $\sqrt{8-(C_{A_0B}(\mathbf{p}_{\Lambda_j}))^2}$, is concave in its argument, $C_{A_0B}(\mathbf{p}_{\Lambda_j})$, and the final equality follows from the operational requirement that averaging $\mathbf{p}_{\Lambda_j}$ over our ignorance of the hidden variable $\lambda$ reproduces the observed behavior, $\mathbf{p}$, i.e., $\mathbf{p}=\sum_{j}p(\Lambda_j)\mathbf{p}_{\Lambda_j}$, and
$C_{A_0B}(\mathbf{p})=\sum_{j}p(\Lambda_j) C_{A_0B}(\mathbf{p}_{\Lambda_j})$. 

Finally, plugging the precondition that the parties observe the post-processed experimental behavior, $\mathbf{p}\equiv \mathbf{p}_{(exp)}^{(\bot \mapsto 1)}$, into \eqref{midway3} yields \eqref{tradeoffeq}.

\end{proof}%

As $\sqrt{8-\left(C_{A_0B}(\mathbf{p}_{(exp)}^{(\bot \mapsto 1)})\right)^2}$ is a monotonically decreasing function of $C_{A_0B}(\mathbf{p}_{(exp)}^{(\bot \mapsto 1)}) \in (2,2\sqrt{2}]$, the inequality \eqref{tradeoffeq} implies that any amount of loophole-free violation of the CHSH inequality \eqref{CHSH} witnessed by $(A_0,B)$, reduces the threshold value of the CHSH expression for $(A_1,B)$. Specifically, when $(A_0,B)$ witness maximally nonlocal correlations, $C_{A_0B}(\mathbf{p}_{(exp)}^{(\bot \mapsto 1)})=2\sqrt{2}$, the threshold value reduces to \emph{zero} \eqref{ClassicalCHSH}, whereas when $(A_0,B)$ fail to violate the CHSH inequality, $(A_1,B)$ must violate the CHSH inequality \eqref{CHSH} on their own, to certify loophole-free nonlocal correlations. 

In what follows, we use Theorem \ref{tradeoff}, as a convenient tool to estimate the critical parameters, $(\eta_{A_1}^*,\nu_{A_1}^*)$, for the loophole-free certification of nonlocal correlations between $(A_1,B)$. 

\textbf{\emph{Analytical estimation of critical parameters:---}} Given a quantum strategy, $\mathcal{Q}$, and the tuple of experimental parameters, $(\eta_B,\eta_{A_0},\nu_{B},\nu_{A_0})$, we retrieve the experimental behavior, $\mathbf{p}_{(exp)}^{(\bot \mapsto 1)}$, as a function of $(\eta_{A_1},\nu_{A_1})$. Now, the critical parameters, $(\eta_{A_1}^*,\nu_1^*)$, are simply the ones for which the inequality \eqref{tradeoffeq} is saturated. 

To demonstrate our methodology, we first consider the asymmetric case wherein, $B$ is placed extremely close to a perfect source, such that $\eta_B=1$, $\eta_{A_0}=\eta$, and all devices have perfect visibility, $\nu_B=\nu_{A_0}=\nu_{A_1}=1$. The parties employ the maximally nonlocal isotropic strategy, $\mathcal{Q}_{iso}$, such that, $C_{A_0B}(\mathbf{p}_{(exp)}^{(\bot \mapsto 1)})=2\sqrt{2}\eta$, where $(A_0,B)$ observe a loophole-hole violation for $\eta\in(\frac{1}{\sqrt{2}},1]$, and $C_{A_1B}(\mathbf{p}_{(exp)}^{(\bot \mapsto 1)})=2\sqrt{2}\eta_{A_1}$, which when plugged into \eqref{tradeoffeq} yields, $\eta^*_{A_1} = \sqrt{1-\eta^2}$. This relation is plotted in FIG. \ref{IsoVersusTiltedFig}, serves to bring up the central insight that the closer to the source $A_0$ is placed, the higher will be its effective detection efficiency, $\eta_{A_0}$, and consequently, the lower will be the critical detection efficiency of $A_1$, $\eta^*_{A_1}$, i.e., the further away $A_1$ can be placed from the source, whilst retaining loophole-free nonlocal correlations with $B$. 

For the symmetric case, wherein $(A_0,B)$ are equidistant from the source, such that, $\eta_B=\eta_{A_0}=\eta$, and all devices have perfect visibility, $\nu_B=\nu_{A_0}=\nu_{A_1}=1$, if the parties employ the isotropic strategy, $\mathcal{Q}_{iso}$, then $C_{A_0B}(\mathbf{p}_{(exp)}^{(\bot \mapsto 1)})=2\sqrt{2}\eta^2+2(1-\eta)^2$, where $(A_0,B)$ observe a loophole-hole violation for $\eta \in(\frac{2}{1+\sqrt{2}},1]$ and $C_{A_1B}(\mathbf{p}_{(exp)}^{(\bot \mapsto 1)})=2\sqrt{2}\eta_{A_1}\eta +  2(1-\eta_{A_1})(1-\eta)$. We plot the values of $\eta^*_{A_1}$, which saturate \eqref{tradeoffeq} against the effective detection efficiency $\eta$ in FIG. \ref{IsoVersusTiltedFig}.

\begin{figure}[t]\centering
 \includegraphics[width=\columnwidth]{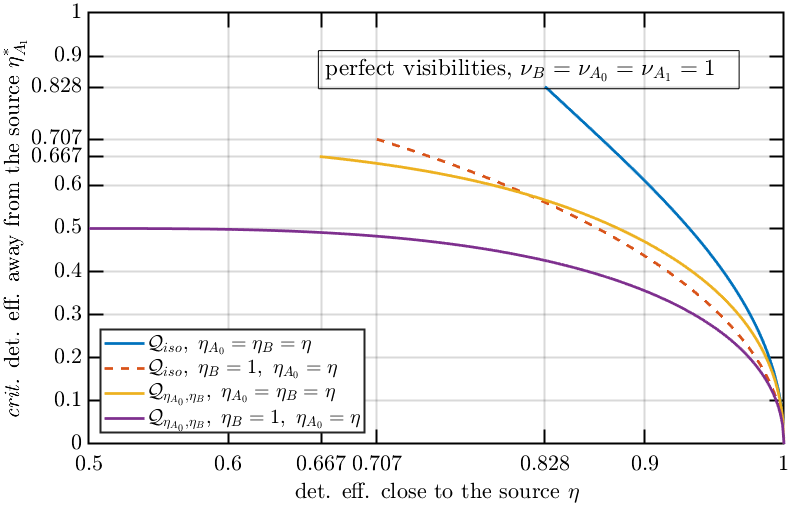}
 \caption{\label{IsoVersusTiltedFig} \textbf{\emph{Isotropic versus tilted strategies:---}} The critical detection efficiency, $\eta^*_{A_1}$, of Alice's measurement device placed away from the source, $A_1$, versus the effective detection efficiency, $\eta$, close to the source, obtained with the analytical trade-off \eqref{tradeoffeq}, when the party's employ the \emph{isotropic} strategy, $\mathcal{Q}_{iso}$, for the symmetric case, $\eta_{A_0}=\eta_B=1$ (top solid blue curve), and for the asymmetric case, $\eta_{A_0}=\eta$, $\eta_B=1$ (dashed orange curve), and when the party's employ the \emph{tilted} strategies, $\mathcal{Q}_{\eta_{A_0},\eta_B}$, for the symmetric case (middle solid yellow curve), and for the asymmetric case (bottom solid purple curve), with perfect visibilities, $\nu_B=\nu_{A_0}=\nu_{A_1}=1$. The critical detection efficiency starts declining after $\eta$ exceeds the respective threshold values: $\eta=\frac{2}{1+\sqrt{2}}\approx0.828$, $\eta=\frac{1}{\sqrt{2}}\approx 0.707$, for the symmetric, and the asymmetric cases, respectively, when the parties employ the isotropic strategy, $\mathcal{Q}_{iso}$, and $\eta=\frac{2}{3}\approx0.67$, $\eta=\frac{1}{2}$, for the symmetric, and the asymmetric cases, when the parties use the tilted strategies, $\mathcal{Q}_{\eta_{A_0},\eta_B}$, respectively. In both symmetric and asymmetric cases, the tilted strategies, $\mathcal{Q}_{\eta_{A_0},\eta_B}$, perform better than the isotropic strategy, $\mathcal{Q}_{iso}$, at minimizing the critical detection efficiency, $\eta^*_{A_1}$.} 
\end{figure}

Up till this point, we have relied exclusively on the isotropic strategy, $\mathcal{Q}_{iso}$. Next, we address whether we can further lower the critical requirements of $A_1$ by employing better quantum strategies.

\textbf{\emph{Optimal quantum strategies:---}} Given the effective detection efficiencies of the measurement devices close to the source, $(\eta_{A_0},\eta_B)$, the parties now use the \emph{tilted strategies}, $\mathcal{Q}_{\eta_{A_0},\eta_B}$, which attain maximum quantum violation of the following tilted CHSH inequality,

\begin{align}\nonumber \label{tilted}
     C^{\eta_{A_0},\eta_B}_{A_0B}(\mathbf{p}) &\equiv  \eta_{A_0}\eta_B  C_{A_0B}(\mathbf{p}) +2(1-\eta_{B})\eta_{A_0}\expval{\hat{a}_{(0,0)}} \\ \nonumber &\hspace{0.1in}+2(1-\eta_{A_0})\eta_B\expval{\hat{b}_{0}}+2(1-\eta_{A_0})(1-\eta_{B}) \\ &\underrel{\mathcal{L}}{\leqslant} 2.
\end{align}
Notice that, given the detection efficiencies, $(\eta_{A_0},\eta_B)$, the strategy, $\mathcal{Q}_{\eta_{A_0},\eta_B}$, attains the maximum loophole-free violation of the CHSH inequality \eqref{CHSH} for $(A_0,B)$, as $C_{A_0B}(\mathbf{p}_{(exp)}^{(\bot \mapsto 1)})=C^{\eta_{A_0},\eta_B}_{A_0B}(\mathbf{p}_{\eta_{A_0},\eta_B})$, where $\mathbf{p}_{\eta_{A_0},\eta_B}\in \mathbb{R}^{32}_{+}$ is the ideal experimental behavior corresponding to the quantum strategy, $\mathcal{Q}_{\eta_{A_0},\eta_B}$. As the threshold value for loophole-free certification of the nonlocal behavior of $A_1$, $\sqrt{8-\left(C_{A_0B}(\mathbf{p}_{(exp)}^{(\bot \mapsto 1)})\right)^2}$ \eqref{tradeoffeq}, is a monotonically decreasing function of $C_{A_0B}(\mathbf{p}_{(exp)}^{(\bot \mapsto 1)})$, the tilted strategies, $\mathcal{Q}_{\eta_{A_0},\eta_B}$, minimize it, and hence, are \emph{optimal} for lowering the critical requirements of $A_1$. In general, one can use the heuristic see-saw semi-definite programming method to numerically find these strategies. In an upcoming work, we describe a much more precise method for obtaining these strategies based on self-testing of the tilted CHSH inequalities \eqref{tilted}. In FIG. \ref{IsoVersusTiltedFig} we plot the resultant curves for $\eta^*_{A_1}$ against the effective detection efficiency, $\eta$, for the asymmetric case, ($\eta_{A_0}=\eta,\eta_B=1$, as well as for the symmetric case, $\eta_{A_0}=\eta_B=\eta$,  demonstrating the advantage of the tilted strategies, $\mathcal{Q}_{\eta_{A_0},\eta_B}$, over the isotropic strategy, $\mathcal{Q}_{iso}$. 

\begin{figure}[t!]\centering
 \includegraphics[width=\columnwidth]{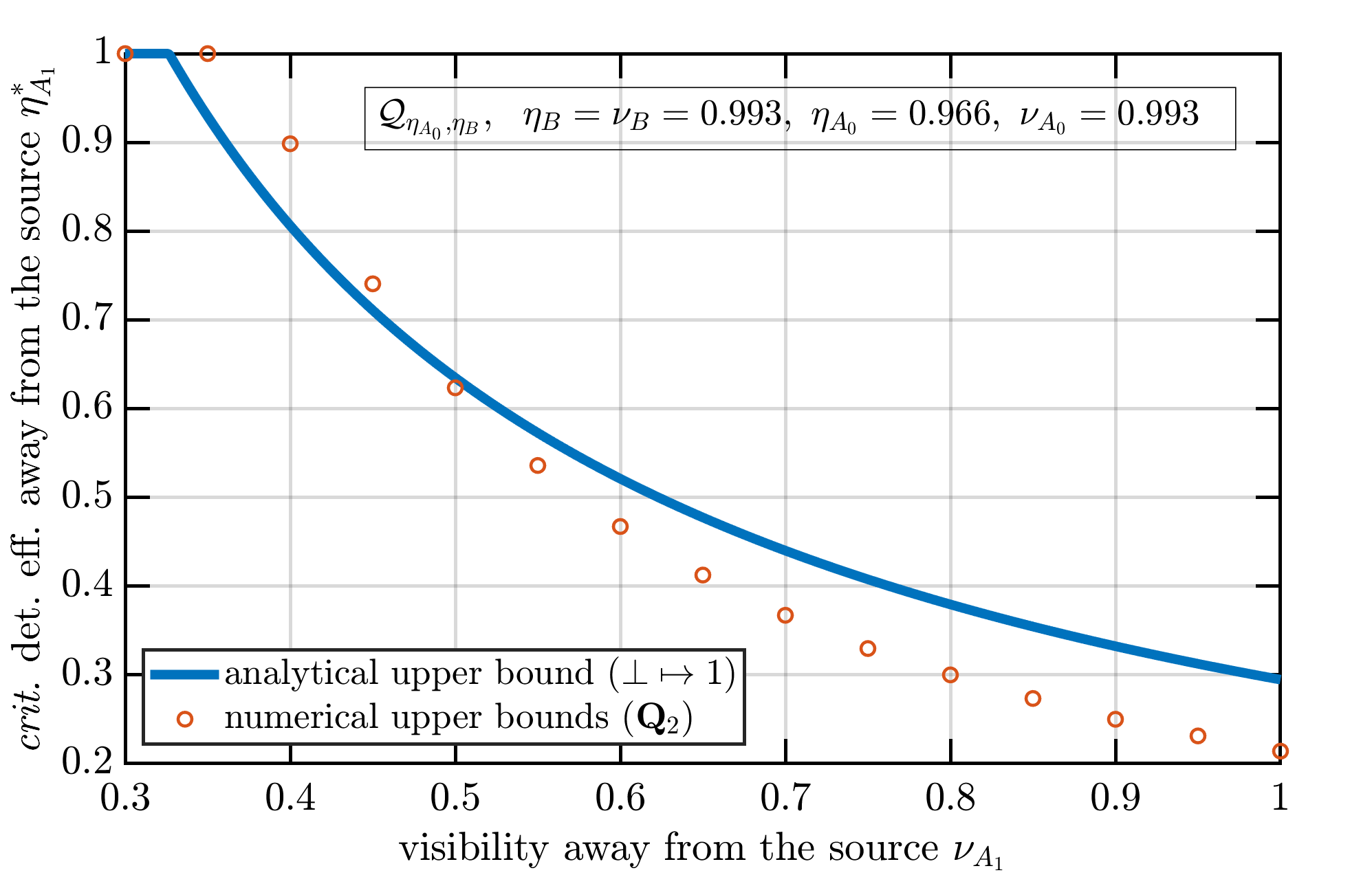}
  \caption{\label{TwoOutcomeVsThreeOutcomeFig} \textbf{\emph{Analytical versus numerical estimation:---}} Upper bounds on the critical detection efficiency, $\eta^*_{A_1}$, of Alice's measurement device placed away from the source, versus the effective visibility between $(A_1,B)$, $\nu_{A_1}$, calculated analytically using \eqref{tradeoffeq}, when the parties use the tilted strategy, $\mathcal{Q}_{\eta_{A_0},\eta_B}$, where $\eta_B=\nu_B=0.993$, $\eta_{A_0}=0.966,\nu_{A_0}=0.993$, and the assignment strategy, $\bot\mapsto+1$ (solid blue curve), and numerically using certifiable randomness as a measure for the nonlocal behavior of ${A_1}$, and the second level $(\mathbf{Q}_2)$ of \emph{Nieto-Silleras} hierarchy of semi-definite programs, while keeping the ``no-click" event as an additional outcome, $\tilde{a},\tilde{b}\in\{\pm1,\bot\}$ (orange circles). The curves demonstrate the advantage of keeping the ``no-click" as an additional outcome, $\tilde{a},\tilde{b}\in\{\pm1,\bot\}$ over the post-processing strategy of assigning a valid outcome to the ``no-click" event, $\bot \mapsto +1$. }
\end{figure}
While Theorem \ref{tradeoff} unitizes the value of the CHSH expression as a measure of nonlocal correlations shared between $(A_1,B)$, and provides a convenient method for the estimation of critical parameters of $A_1$, its applicability is limited to the simplest Bell scenario, as well as to the choice of the post-processing strategy, specifically, assigning a pre-determined valid outcome whenever the measurement device fails to detect the incoming quantum system (the ``no-click" event). We now describe a much more broadly applicable numerical method for the estimation of the critical parameters of individual measurement devices in generic network Bell scenarios.

{\textbf{\emph{A versatile numerical tool:---}}} For generic Bell scenarios, a more versatile measure of the nonlocal behavior of an individual measurement device, for instance, $A_1$, which takes into the raw three-outcome experimental behavior, $\mathbf{p}_{(exp)}$, is the amount of \emph{certifiable randomness}, $H_{min}(G_{x}(\mathbf{p}_{(exp)}))$, where $H_{min}(\cdot)$ is the min-entropy, and $G_{x}(\mathbf{p}_{(exp)})$ is maximum guessing probability \cite{Bancal_2014}, 
\begin{equation} \label{eq:program}
\begin{split}
G_{x}(\mathbf{p}_{(exp)}) = \underset{\mathbf{p}_{\tilde{a}}}{\ \ \max\ \ } & \sum_{\tilde{a}\in\{\pm1,\bot\}} p_{\tilde{a}}(\tilde{a}|x,1)\\
\text{s.t.\ \ } &  \sum_{\tilde{a}\in\{\pm1,\bot\}} \mathbf{p}_{\tilde{a}}= \mathbf{p}_{(exp)}\\
&  \forall \tilde{a}\in\{\pm1,\bot\}: \ \mathbf{p}_{\tilde{a}}\in \mathbf{Q},
\end{split}
\end{equation}
where $p_{\tilde{a}}(\tilde{a}|x,1)=\sum_{\tilde{b}\in\{\pm1,\bot\}}p_{\tilde{a}}(\tilde{a},\tilde{b}|(x,1),y)$, $\{\mathbf{p}_{\tilde{a}}\in\mathbb{R}^{72}_+\}$ are convex decompositions of the raw experimental behavior $\mathbf{p}_{(exp)}\in\mathbb{R}^{72}_+$, wherein we have absorbed the convex coefficients into the respective decompositions, and $\mathbf{Q}$ is the convex set of quantum behaviors \footnote{All behaviors realizable with quantum resources, specifically, by performing local measurements on a shared entangled state.}, $\mathbf{p}\equiv \{p(\tilde{a},\tilde{b}|(x,i),y)\}\in\mathbb{R}^{72}_+$. 

Observe that, for any behavior possessing a local hidden variable explanation for $A_1$ of the form \eqref{A1LHVGenBellExp2}, we can, without loss of generality, take the response schemes to be deterministic, i.e., $p^{A}_\lambda(a|x,i=1)\in\{0,1\}$. This observation in-turn implies that for an experimental behavior $\mathbf{p}_{(exp)}$ possessing a local hidden variable explanation for Alice's device $A_1$ similar to \eqref{A1LHVGenBellExp2}, the outcomes of $A_1$ can always be perfectly guessed if one has access to the local hidden variable $\lambda$, such that $G_{x}(\mathbf{p}_{(exp)})=1$. On the contrary, if the outcomes of $A_1$ cannot be guessed perfectly (prior to the specific round of the experiment), i.e., the guessing probability optimization yields $G_{x}(\mathbf{p}_{(exp)})<1$, no local hidden variable explanation for the experimental behavior, similar to \eqref{A1LHVGenBellExp2}, exists. Consequently, the critical parameters for witnessing loophole-free nonlocality, $(\eta^*_{A_1},\nu^*_{A_1})$, correspond to the threshold values of $(\eta_{A_1},\nu_{A_1})$, for which, $G_{x}(\mathbf{p}_{(exp)}) < 1$.

This measure is particularly well-suited for generic network Bell scenarios entailing many spatially separated parties, wherein we are interested in the certification of the nonlocal behavior of individual measurement devices when other measurement devices may or may not be already witnessing loophole-free nonlocality. Moreover, apart from the direct relevance to DI randomness certification, $H_{min}(G_{x}(\mathbf{p}_{(exp)}))$, constitutes a crucial ingredient in the security proofs of more complex DI cryptography protocols. 
 \begin{figure}\centering
 \includegraphics[width=\columnwidth]{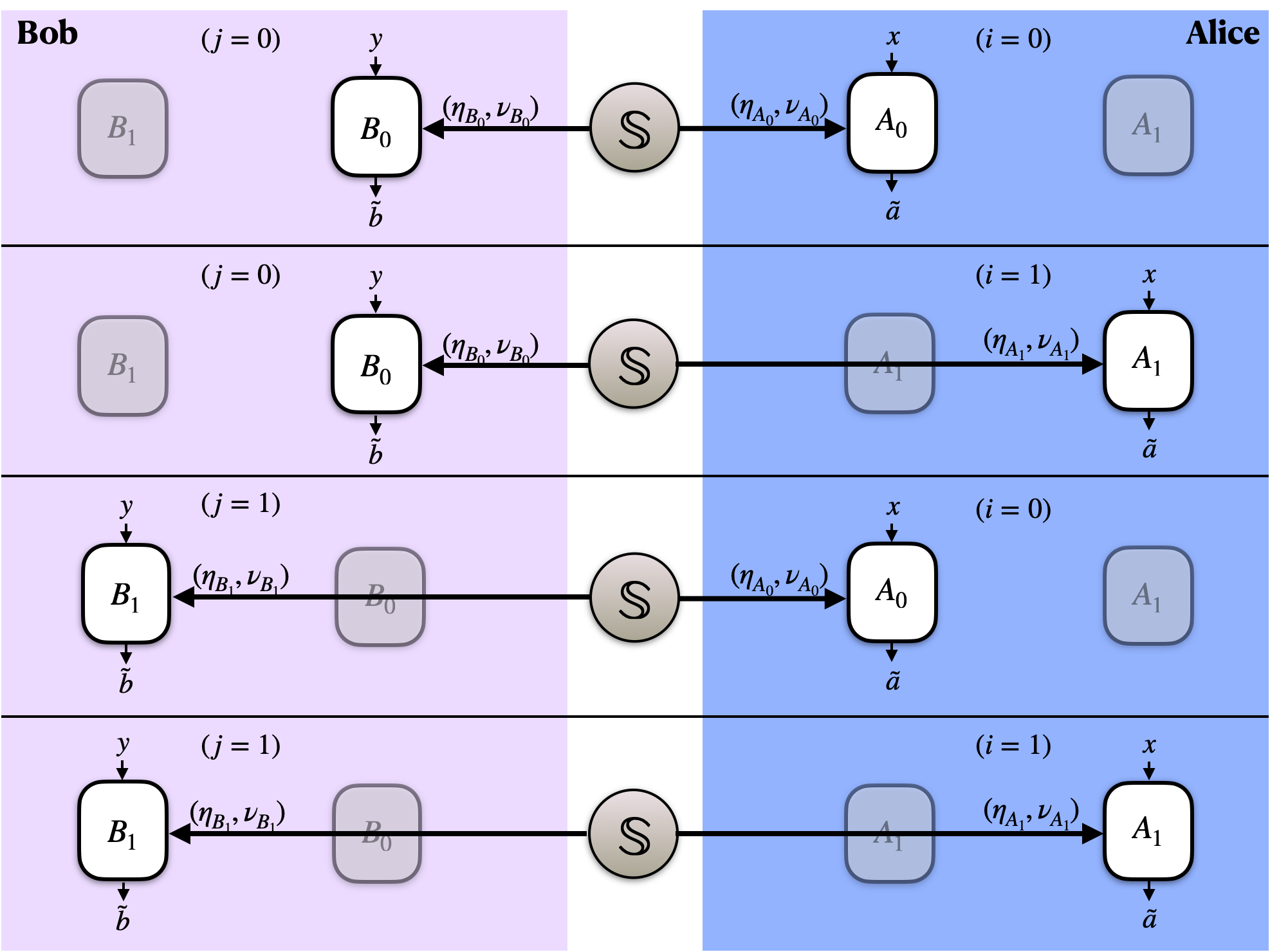}
 \caption{\label{TwoAliceTwoBob} \textbf{\emph{Symmetric generalization:---}} The graphic is a schematic depiction of the next generalization to the generalized or routed Bell experiment depicted in FIG. \ref{Topology}. Along with Alice, in each round of the experiment, Bob chooses the location of his measurement device, $B_j$, based on a randomly chosen input bit, $j\in\{0,1\}$. When $j=0$, Bob places his measurement device, $B_0$, close to the source with an effective detection efficiency and visibility, $(\eta_{B_0},\nu_{B_0})$. When $j=1$, he places her measurement device, $B_1$, further away from the source, with a lower detection efficiency and visibility, $(\eta_{B_1},\nu_{B_1})$, such that $\eta_{B_1}<\eta_{B_0}$, $\nu_{B_1}<\nu_{B_0}$. Our results imply that in rounds with $i=j=1$, loophole-free nonlocal correlations can be certified between the measurement devices, $(A_1, B_1)$, placed arbitrarily far away from each other.}
 \end{figure}
 
The optimization problem \eqref{eq:program} is extremely arduous to solve because of the constraint $\mathbf{p}_{\tilde{a}}\in \mathbf{Q}$. Instead, to retrieve progressively tightening upper-bounds on the maximum guessing probability, $G_{x}(\mathbf{p}_{(exp)})$, and the estimation of the critical parameters $(\eta^*_{A_1},\nu^*_{A_1})$, we employ the \emph{Nieto-Silleras} hierarchy of semi-definite programs \cite{Nieto_Silleras_2014,Bancal_2014}, which relaxes the constraint $ \mathbf{p}_{\tilde{a}}\in \mathbf{Q}$, to $\mathbf{p}_{\tilde{a}}\in \mathbf{Q}_L$, where $\mathbf{Q}_L$ are the convex relaxations of the quantum set, $\mathbf{Q}$, corresponding to \emph{Navascu\'es--Pironio--Ac\'in hierarchy}, $L\in\mathbb{N}_+$ denotes the level of the relaxation, such that, $\mathbf{Q}\subset\mathbf{Q}_{L+1}\subset\mathbf{Q}_{L}$ for all $L\in\mathbb{N}_+$, and $\lim_{L\to\infty}\mathbf{Q}_L=\mathbf{Q}$ \cite{Navascues_2008}. In FIG. \ref{TwoOutcomeVsThreeOutcomeFig} we plot the critical detection efficiency of $A_1$, $\eta^*_{A_1}$ versus its visibility $\nu_{A_1}$, for an experimentally relevant case, while keeping the ``no-click" as an additional outcome, $\tilde{a},\tilde{b}\in\{\pm1,\bot\}$, and using the second level of the \emph{Nieto-Silleras} hierarchy, i.e., with the relaxed constraint $\mathbf{p}_{\tilde{a}}\in \mathbf{Q}_2$. In FIG. \ref{TwoOutcomeVsThreeOutcomeFig}, we also plot the corresponding curve retrieved analytically using \eqref{tradeoffeq} and the assignment strategy, $\bot\mapsto+1$. The plot serves to demonstrate the advantage of the numerical technique, and in particular, of keeping the entire raw experimental behavior over the analytical technique and the post-processing strategy of assigning a valid outcome to the ``no-click" event.

\textbf{\emph{Discussion:---}} In this work, we formulated a novel scheme to overcome the most significant impediment in realizing long-distance loophole-free nonlocal correlations, namely, the detection loophole. Our scheme exploits the properties of short-distance loophole-free nonlocal correlations, which can be readily attained in present-day Bell experiments to extend them to longer distances. Specifically, 
we considered Bell experiments wherein the involved parties randomly choose the location of their measurement devices in each round. To demonstrate the considered effect, we stuck to the most straightforward generalization of the Bell-CHSH experiment, wherein only Alice randomly chooses the location of her measurement devices based on an input $i\in\{0,1\}$. However, our approach can be applied to more complex scenarios as well. For instance, consider the case wherein, along with Alice, Bob randomly chooses the location of his measurement device, $B_j$, based on an input $j\in\{0,1\}$, in each round of the experiment (FIG \ref{TwoAliceTwoBob}). When $i=j=0$, both parties place their devices close to the source and witness strong loophole-free nonlocal correlations. Consequently, in the rounds when a party places their measurement device close to the source while the other places their device away from the source, i.e. when either $i=0,j=1$ or $i=1,j=0$, loophole-free nonlocal correlations can be certified at arbitrarily low detection efficiency and visibility. This observation, in turn, enables the certification of loophole-free nonlocal correlations between measurement devices placed arbitrarily far away from the source and each other when $i=j=1$. 

Physically moving the measurement device during the experiment is arduous and impractical. Therefore, experiments aimed at showcasing the considered effect could employ \emph{relay-switches}, which alter the path of the quantum system transmitted from the source actively, which is to say, based on inputs from the involved parties, in each round of the experiment. We call such Bell experiments ``routed" Bell experiments, specifically for the simple case depicted in FIG. \ref{Topology}, such a relay switch can be placed between the source and the measurement device $A_1$. In each round of the experiment, Alice transmits her choice of $i\in\{0,1\}$ to the relay switch, which then transmits the quantum system from the source to either the measurement device, $A_0$, placed close to the source or $A_1$, placed further away, based on the input from Alice, $i$. 

The generalized or routed Bell experiments introduced in this work are closely related to the EPR steering scenarios \cite{Steering}; however, the ``trust structure" is different. Specifically, in bipartite steering scenarios, the  measurement device of one of the spatially separated parties, referred to as the ``steering party", is completely characterized or ``trusted." On the other hand, our treatment of the generalized or routed Bell experiments is completely device-independent, i.e., all devices along with the source remain completely uncharacterized or ``untrusted". We use the fact that the devices close to the source witness nonlocal correlations, specifically, a violation of CHSH inequality, to characterize Bob measurement device $B$ subsequently. We then use this characterization of Bob's measurement device to derive operational consequences of local hidden variable models for $A_1$, such as \eqref{tradeoffeq}. Finally, we demonstrate the operational quantum violation of these consequences.  

Besides increasing the distance over which loophole-free nonlocal correlations can be sustained, our scheme enables the certification of nonlocal behavior of off-the-shelf measurement devices. Finally, it follows from the proof of Theorem \ref{IdealCase} that the parties can extend extremal loophole-free nonlocal correlations to an arbitrarily large number of additional measurement devices placed away from the source. These observations together enable novel applications, such as a central hub equipped with expensive state-of-the-art measurement apparatus witnessing strong loophole-free nonlocal correlations and distributing them to an arbitrary number of remotely located commercial off-the-shelf measurement devices, making loophole-free nonlocality much more broadly accessible. To summarize, we anticipate our findings to significantly accelerate the advent of DI information processing as a near-term commonplace technology.

Apart from the direct application of the routed Bell experiments introduced here to Device Independent Randomness Certification with very inefficient detectors, the fact that \emph{local hidden variable} explanations of the form \eqref{LHVGenBellExp},\eqref{A1LHVGenBellExp2} are isomorphic to quantum explanations entailing local measurements on separable states implies that the violation of their operational consequences, such as the violation of the trade-off \eqref{tradeoffeq}, in routed Bell experiments can enable Device Independent Entanglement Certification and Quantification \cite{DIEQ} at arbitrarily large distances. Moreover, the scheme proposed in this work can implement DIQKD protocols, enabling remote measurement devices to share a secret key \footnote{A notable DIQKD protocol involving local Bell tests to avoid the detection loophole was proposed in \cite{PhysRevX.3.031006}. However, the protocol utilized an intricate entanglement swapping protocol as a sub-routine, making it vulnerable to noise.}. Although naive DIQKD protocols based on this scheme are secure against passive source-based attacks and eavesdroppers placed between the source and Alice's closest measurement device $A_0$, they are susceptible to active device-controlling attacks by eavesdroppers placed between Alice's measurement devices $(A_0, A_1)$. Specifically, in the Appendix, we demonstrate that for a measurement device with $N$ inputs and $N'$ outputs, there exist two distinct active attacks which render any DIQKD protocol insecure whenever the detection efficiency of such a measurement device is below $\min\{\frac{1}{N},\frac{1}{N'}\}$. In our set-up this translates to a minimum detection efficiency of $\frac{1}{2}$ for all measurement devices for completely secure DIQKD. Therefore, a crucial open question remains whether the generalized or routed Bell experiments introduced in this work can be used to implement improved DIQKD protocols with lower critical detection efficiency and visibility requirements than standard DIQKD protocols. 

{Furthermore, while the bipartite quantum strategies considered here are well suited to the routed Bell experiments, tripartite no-signaling quantum strategies are an interesting alternative natural generalization to the local hidden variable models described in \eqref{A1LHVGenBellExp2}. Specifically, probability distributions explainable via such local hidden variable models can be interpreted as those stemming from local measurements performed on a shared tripartite classical-quantum-quantum state. It then follows that the most obvious modeling of the fully quantum case is promoting the classical register of Alice's distant device to a quantum register such that the devices measure a tripartite quantum-quantum-quantum state. However, such tripartite quantum strategies have an additional implicit no-signaling constraint between Alice's devices compared to the bipartite quantum strategies considered in the article. Consequently, even the most straightforward experimentally implementable quantum strategy for the generalized or routed Bell experiments wherein the source shares (potentially noisy) maximally entangled states to the devices cannot be accounted for by such tripartite quantum strategies, nor could such strategies violate the trade-off \eqref{tradeoffeq}. Nevertheless, it is interesting to explore in more detail the role of tripartite quantum strategies, especially strategies involving genuine tripartite entanglement in routed Bell experiments.} Finally, it would be interesting to replace the guessing probability maximization \eqref{eq:program} with the minimization of conditional von Neumann entropies \cite{Brown2021} and investigate the behavior in the context of the generalized or routed Bell experiments introduced in this work.

\section*{Note Added}
After completion and communication of the first version of our manuscript we were informed of an independent work \cite{lobo2023certifying} which resolves the open question of enabling DIQKD protocols based on the routed Bell experiments introduced in this article with lower critical requirements than standard DIQKD schemes. In particular, the novel scheme employs jointly-measurable measurements for Alice's device $A_1$ as classical models, instead of the \emph{local hidden variable} models invoked here.

\begin{acknowledgments}
We are grateful to Ekta Panwar, Nicolás Gigena, and Piotr Mironowicz, Jef Pauwels, Tam\'{a}s V\'{e}rtesi, Renato Renner, Stefano Pironio \& Marek \.{Z}ukowski for enlightening discussions. This work was partially supported by the Foundation for Polish Science (IRAP project, ICTQT, contract No. MAB/218/5, co-financed by EU within the Smart Growth Operational Programme) and partially supported by the Institute of Information \& Communications Technology Planning \& Evaluation (IITP) grant funded by the Korean government (MSIT) (No.2022- 0-00463, Development of a quantum repeater in optical fiber networks for quantum internet). MP also acknowledges financial support from QuantERA, an ERA-Net co-fund in Quantum Technologies (www.quantera.eu), under project eDICT (contract No. Quantera/2/2020).
The numerical optimization was carried out using \href{https://ncpol2sdpa.readthedocs.io/en/stable/index.html}{Ncpol2sdpa} \cite{wittek2015algorithm}, \href{https://yalmip.github.io/}{YALMIP} \cite{Lofberg2004}, and \href{https://www.mosek.com/documentation/}{MOSEK} \cite{mosek}.
\end{acknowledgments}

\
\clearpage

\appendix
\onecolumngrid
\section{Active attacks on DIQKD}
In this section, building up on the results contained in \cite{MassarPhysRevA.68.062109,Tomamichel2012}, we retrieve a generic lower bound on the threshold detection efficiency $\eta>\min\{\frac{1}{N},\frac{1}{N'}\}$ of any spatially separated measurement device $N$ settings and $N'$ outcomes for the security of any DIQKD protocol. To this end, we describe the following two distinct active attacks,

\begin{attack}[\textbf{\emph{First lower bound on detector efficiency for QKD}}] \label{bound1}
Let the Bell experiment involve two parties: Alice with $N$ choices of setting $x\in[N]$ and $N'$ outcome alphabet $a\in[N']$, and Bob with $M$ settings $y\in[M]$ and $M'$ outcomes $b\in[M']$. There exists an active attack which allows the eavesdropper to reproduce any no-signalling probability distribution as long as the detector efficiencies satisfy $\eta_A\leq \frac{1}{N}$ and $\eta_B\leq \frac{1}{M}$.
\end{attack} 
\begin{proof}
Let the probability distribution that the honest, non-local devices would observe with perfect detectors be $p_{NL}(a,b|x,y)$, where $a\in[N']$ and $b\in[M']$ are the outcomes of Alice's and Bob's measurements and $x\in[N]$ and $y\in[M]$ their respective settings. Taking into the account detectors' inefficiencies the observed probability distribution is:
\begin{align} \nonumber
p(a,b|x,y)=\eta_A \eta_B p_{NL}(a,b|x,y)
\\ \nonumber
p(a,\bot |x,y)=\eta_A (1-\eta_B) p^{(A)}_{NL}(a|x)
\\ \nonumber
p(\bot,b |x,y)=(1-\eta_A) \eta_B p_{NL}^{(B)}(b|y)
\\ \label{distribution}
p(\bot,\bot |x,y)=(1-\eta_A) (1-\eta_B),
\end{align}
where $\bot$ is the symbol assigned to the outcome if the detector doesn't produce any outcome, and $p^{(A)}_{NL}(a|x)$ and $p^{(B)}_{NL}(b|y)$ are the marginal probability distributions for Alice and Bob respectively.

Let's now consider a model which distributes four classical variables $\lambda_A,\lambda_B,\lambda_X,\lambda_Y$ to the parties. $\lambda_X$ and $\lambda_Y$ are uniformly distributed and have values corresponding to the possible choices of settings, i.e. $p(\lambda_X=x)=\frac{1}{N}$ and $p(\lambda_Y=y)=\frac{1}{M}$. The remaining two have the alphabets corresponding to the outcomes and are sampled using the distribution $p(\lambda_A =a,\lambda_B=b|\lambda_X=x,\lambda_Y=y)=p_{NL}(a,b|x,y)$. This model, for instance, could be realized by an eavesdropper intercepting and measuring the transmission from the source.  The outcomes are generated locally with the following strategy. If $x=\lambda_X$ then $a=\lambda_A$ otherwise $a=\bot$, if $y=\lambda_Y$ then $b=\lambda_B$ otherwise $b=\bot$. It is easy to check that the following strategy reproduces probability distribution (\ref{distribution}) with $\eta_A=\frac{1}{N}$ and $\eta_B=\frac{1}{M}$.
\end{proof}

\begin{attack}[\textbf{\emph{Second lower bound on detector efficiency for QKD}}] \label{bound2}
Let the experiment involve two parties: Alice with $N$ choices of setting $x\in[N]$ and $N'$ outcome alphabet $a\in[N']$, and Bob with $M$ settings $y\in[M]$ and $M'$ outcomes $b\in[M']$. There exists an active attack which allows the eavesdropper to reproduce any no-signalling probability distribution as long as the detector efficiencies satisfy $\eta_A\leq \frac{1}{N'}$ and $\eta_B\leq \frac{1}{M'}$.
\end{attack} 
\begin{proof}
Unlike the previous case, the source distributes now a genuine non-local probability distribution \emph{and} two hidden variables $\lambda_A,\lambda_B$ to the parties. $\lambda_A$ and $\lambda_B$ are uniformly distributed and have values corresponding to the possible outcomes, i.e. $p(\lambda_A=a)=\frac{1}{N'}$ and $p(\lambda_B=b)=\frac{1}{M'}$. 

The outcomes are generated locally with the following strategy. The devices make measurements on the non-local system and record the outcomes $a$ and $b$ for postprocessing. If $a=\lambda_A$ the outcome $a$ is returned, otherwise the returned outcome is $a=\bot$. Analogous strategy is employed by Bob's device. 

We notice that the eavesdropper with access to the hidden variables knows the outcomes of both parties whenever they are different from $\bot$, which in QKD protocols is enough to know all the generated key. The efficiencies generated by this strategy are  $\eta_A=\frac{1}{N'}$ and $\eta_B=\frac{1}{M'}$, respectively.
\end{proof}
We note that in both of these attacks the Eavesdropper has to actively intercept and measure the transmission from the source to obtain the distributions $p(\lambda_A =a,\lambda_B=b|\lambda_X=x,\lambda_Y=y)=p_{NL}(a,b|x,y)$ for the first attack, and non-local distribution in the second attack, justifying the qualifier ``active". As a consequence of these attacks, we obtain the desired lower on the detection efficiency $\min\{\frac{1}{N},\frac{1}{N'}\}$ of a spatially separated measurement device with $N$ inputs and $N'$ outputs for the security of DIQKD protocols.

\bibliographystyle{plain}
%

\end{document}